\documentclass[letterpaper, 10 pt, conference]{ieeeconf}  

\IEEEoverridecommandlockouts                              

\usepackage{amssymb}
\usepackage{graphicx}
\usepackage{epstopdf}
\usepackage{amsmath}
\usepackage{mathtools, cuted}
\usepackage{multirow}
\usepackage{subfigure}
\usepackage{array}
\usepackage{bm}
\usepackage{url}

\usepackage{tabularx}
\usepackage{stackengine}
\setstackEOL{\\}

\newtheorem{problem}{Problem}
\newtheorem{definition}{Definition}

\newtheorem{lemma}{Lemma}
\newtheorem{corollary}{Corollary}

\newtheorem{proposition}{Proposition}

\title{\LARGE \bf Detection of Hidden Attacks on Cyber-Physical Systems from Serial Magnitude and Sign Randomness Inconsistencies
}
\author{Paul J Bonczek and Nicola Bezzo
\thanks{Paul J Bonczek and Nicola Bezzo are with the Charles L. Brown Department of Electrical and Computer Engineering, and Link Lab, University of Virginia, Charlottesville, VA 22904, USA. Email: {\tt \{pjb4xn, nb6be\}@virginia.edu}}
}

\newcommand*{\N}{\mathbb{N}}
\newcommand*{\R}{\mathbb{R}}
\newcommand*{\E}{\mathbb{E}}
\newcommand*{\PP}{\mathbb{P}}

\renewcommand{\baselinestretch}{0.95}
\begin{document}

\maketitle
\thispagestyle{empty}
\pagestyle{empty}

\begin{abstract}

Stealthy false data injection attacks on cyber-physical systems (CPSs) introduce erroneous measurement information to on-board sensors with the purpose to degrade system performance. An intelligent attacker is able to leverage knowledge of the system model and noise characteristics to alter sensor measurements while remaining undetected. To achieve this objective, the stealthy attack sequence is designed such that the detector performs  similarly in the attacked and attack-free cases. Consequently, an attacker that wants to remain hidden will leave behind traces of inconsistent behavior, contradicting the system model. To deal with this problem, we propose a runtime monitor to find these inconsistencies in sensor measurements by monitoring for \textit{serial inconsistencies} of the detection test measure. Specifically, we employ the chi-square fault detection procedure to monitor the magnitude and signed sequence of its chi-square test measure. We validate our approach with simulations on an unmanned ground vehicle (UGV) under stealthy attacks and compare the detection performance with various state-of-the-art anomaly detectors.

\end{abstract}
\section{Introduction} \label{sec:introduction}

Modern cyber-physical systems (CPSs) have been the targets of malicious cyber-attacks due to their growing unsupervised, autonomous capabilities and many entry points to implement an attack. Their expanded complexities supported by an increased number of sensors and computers allow for autonomous capabilities in navigation, warehouse logistics, surveillance, warfare, and industrial operations. With a growing number of vulnerable access points for attackers on increasingly impactful systems in our society, it is crucial to provide tighter security measures to ensure proper performance and safety.

An intelligent attacker is able to implement a malicious attack sequence to manipulate the system of interest, all while remaining undetected. The execution of such a stealthy attack allows the intelligent attacker to degrade system performance and potentially cause damage to the unknowingly compromised system. Previously demonstrated attacks of this nature include cases like: the GPS spoofing of a vessel \cite{YachtSpoof}, different sensor and communication attacks on vehicle technologies \cite{miller2015remote}, and the infamous Stuxnet attack \cite{Stuxnet}.

In order to repel these stealthy attacks, detection algorithms are designed to find compromised system's components to maintain safe operation. Intelligent attackers, in turn, resort to new methods to hide and deceive on-board control systems and their anomaly detection counterparts. While such attack vectors are less effective since the attacker needs to maintain a low profile, performance degradation can be still accomplished if the attack is able to remain undetected.

We note, however, that in general an attacker needs to create inconsistent behavior with respect to the known model in order to hijack a system. In this work, we consider the chi-square detection scheme \cite{BadData} that generates a scalar quadratic {\em test measure} for attack detection. This test measure is extracted from the sum of squares of the residual vector --- defined as the vector of differences between sensor measurements and the state prediction. To detect inconsistencies, we monitor the serial behavior of the test measure difference; specifically, we observe the characteristics of consecutive test measures throughout a sequence of measurement data and compare them to an expectation extracted from prior knowledge about the system model. Our proposed Serial Detector is then designed to generate an alarm rate at runtime for detection purposes to discover inconsistent magnitude and sign behavior due to deceptive sensor attacks.

The main objective of this work is to find intelligent sensor attack sequences that deliberately attempt to remain hidden from conventional detection techniques in noisy dynamical systems. The contribution of this paper is twofold: 1) we propose the Serial Detector to monitor inconsistent magnitude and sign behavior of the test measure difference within a system employing a chi-square fault detection procedure, and 2) we characterize a worst-case scenario that an attacker can exploit to remain undetected from our proposed detector.

\subsection{Related Work}
\label{sec:Related Work}

The field of CPS security has garnered much interest in the robotics, controls, and computer science communities to protect critical systems. In recent literature, several procedures that analyze components of the residual in the control system feedback, namely the $\chi^2$ test measure \cite{BadData}, for attack detection have also been exploited. For example, the model-based Cumulative Sum (CUSUM) procedure proposed in \cite{CUSUM_Journal} leverages the known noise characteristics of the system model and sequentially sums the test measure to detect changes within its distribution. In \cite{CST2}, authors included a coding matrix to the sensor outputs, unknown to attackers, to detect stealthy attacks through an iterative optimization algorithm for solving a transformation matrix. Other similar detection procedures, such as in \cite{watermarking}, leverage watermarking of the control inputs to discover stealthy attacks. 

Our recent works on attack detectors that monitor for non-random residual behavior have enabled the ability to find previously undetectable attacks when compared to conventional detection procedures. In \cite{Paul_ACC}, a windowed detector leveraging the Wilcoxon-Signed Rank \cite{Wilcoxon1} and Serial Independence Runs \cite{serial_test} tests was proposed to find non-random patterns over a sequence of sensor data. Similarly, in \cite{Paul_IFAC} we characterized the Cumulative Sign (CUSIGN) detector with the purpose of finding non-random signed residual behavior by checking for changes in probability of the signed values. In this paper, we expand on these previous works by developing a runtime monitor for both non-random magnitude and sign behaviors, further strengthening detection capabilities.

The remainder of this work is organized as follows. In Section \ref{sec:preliminaries} we begin with system and estimation models along with the problem formulation, followed by the description of our Serial Detection framework in Section \ref{sec:framework}. We provide an attack analysis in Section \ref{sec:Undetected_Attacks} to expose a worst-case scenario for Serial Detection. Finally, in Section \ref{sec:Results} we present numerical simulations to demonstrate the performance of our proposed detector and compare with three state-of-the-art algorithms, before discussing our conclusions in Section \ref{sec:conclusion}.
\section{Preliminaries \& Problem Formulation} \label{sec:preliminaries}

In this work we consider discrete-time linear time-invariant (LTI) systems in the following form:
\begin{equation}\label{eq:system1}
\begin{split}
\bm{x}_{k+1} &= \bm{A} \bm{x}_k + \bm{B} \bm{u}_k + \bm{\nu}_k, \\
	\bm{y}_k&=\bm{C} \bm{x}_k + \bm{\eta}_k ,
\end{split}
\end{equation}
where the state vector $\bm{x}_k \hspace{-1pt} \in \hspace{-.5pt} \R^{n}$, $k \hspace{-1pt} \in \hspace{-.5pt} \N$ evolves due to the discrete-time state transition and input matrices $\bm{A} \in \R^{n\times n}$ and $\bm{B} \in \R^{n\times m}$, control input $\bm{u}_k \in \R^m$, and additive i.i.d. zero-mean Gaussian process uncertainty $\bm{\nu}_k = \mathcal{N}(0,\bm{Q}) \in \R^n$ described by the covariance matrix $\bm{Q} \in \R^{n\times n}, \bm{Q} \geq 0$. The output vector $\bm{y}_k \in \R^{s}$ represents the measured system states with additive i.i.d. zero-mean Gaussian measurement uncertainty $\bm{\eta}_k = \mathcal{N}(0,\bm{R}) \in \R^s$ with covariance matrix $\bm{R} \in \R^{s\times s}, \bm{R} \geq 0$ that provides measurements to $s$ sensors.

We consider sensor measurements $\bm{y}_k$ that can be altered due to an additive attack vector $\bm{\xi}_k \in \R^s$, which results in an attacked output measurement vector described by
\begin{equation} \label{eq:attacked_output}
    \Tilde{\bm{y}}_k = \bm{y}_k + \bm{\xi}_k = \bm{C} \bm{x}_k + \bm{\eta}_k + \bm{\xi}_k \in \R^s.
\end{equation}

During operations, a steady state Kalman Filter with gain matrix $\bm{L} \in \R^{n \times s}$ is implemented to provide a state estimate $\hat{\bm{x}}_k \in \R^n$ in the form:
\begin{equation}\label{eq:Kalman}
\begin{split}
	\hat{\bm{x}}_{k+1} &= \bm{A} \hat{\bm{x}}_k + \bm{B} \bm{u}_k + \bm{L}(\Tilde{\bm{y}}_k - \bm{C}\hat{\bm{x}}_k), \\
	\bm{L} &= \bm{P}\bm{C}^{\mathsf{T}}(\bm{C}\bm{P}\bm{C}^{\mathsf{T}} + \bm{R})^{-1}.
\end{split}
\end{equation}

The gain matrix $\bm{L}$ results in a minimal steady state estimation error covariance matrix $\bm{P} = \E[\bm{e}_k \bm{e}_k^{\mathsf{T}} ]$, where $\bm{e}_k = \bm{x}_k - \hat{\bm{x}}_k$ is the estimation error. The measurement residual vector is defined as
\begin{equation}
\label{eq:Residual}
	\bm{r}_k = \Tilde{\bm{y}}_k - \bm{C}\hat{\bm{x}}_k = \bm{C}\bm{e}_k + \bm{\eta}_k \in \R^s,
\end{equation}
with an expected residual covariance matrix, in attack-free conditions (i.e., $\bm{\xi}_k = \bm{0}$), described by:
\begin{equation}
\label{eq:Residual_Covariance}
	\bm{\Sigma} = \E[\bm{r}_{k}\bm{r}_{k}^{\mathsf{T}}]  = \bm{C}\bm{P}\bm{C}^{\mathsf{T}} + \bm{R} \hspace{2pt} \in \R^{s \times s}.
\end{equation}

For the measurement residual, we test two different hypotheses: The null hypothesis $\mathcal{H}_0$ for nominal scenario (attack-free) and the alternative hypothesis $\mathcal{H}_a$ where attacks are present. Formally, the hypotheses are written as
\begin{equation} \label{eq:residual_hypothesis}
    \mathcal{H}_0: \bigg\{ \hspace{-2pt} \begin{array}{l}
    \begin{aligned}
	\E[\bm{r}_k] &= 0, \\
    \E[\bm{r}_k \bm{r}_k^{\mathsf{T}}] &= \bm{\Sigma},
    \end{aligned}
    \end{array} \;\;\; \mathcal{H}_a: \bigg\{ \hspace{-2pt} \begin{array}{l}
    \begin{aligned}
	\E[\bm{r}_k] &\ne 0, \text{and/or} \\
    \E[\bm{r}_k \bm{r}_k^{\mathsf{T}}] &\ne \bm{\Sigma}.
    \end{aligned}
    \end{array}
\end{equation}

In this work, we consider a single detector to monitor the system for sensor attacks by way of the chi-square detector, which produces a scalar quadratic \textit{test measure} $z_k$ by
\begin{equation} \label{eq:test_measure}
    z_k = \bm{r}_k^{\mathsf{T}} \bm{\Sigma}^{-1} \bm{r}_k \in \R_{\geq 0}.
\end{equation}

In the absence of attacks, the measurement residual is an $s$-dimensional vector of normally distributed random variables $\bm{r}_k \sim \mathcal{N}(0,\bm{\Sigma})$, satisfying the null hypothesis $\mathcal{H}_0$ in \eqref{eq:residual_hypothesis}. The test measure $z_k$ is then expected to be a random variable that follows a chi-square distribution with $s$ degrees of freedom, i.e. $z_k \sim \chi^2(s)$, that follows:
\begin{equation}
    \label{eq:Expected_z}
	\E[z_{k}] = s, \;\;\; \text{Var}[z_{k}] = 2s.
\end{equation}

\subsection{Undetected Attacks}
\label{sec:HiddenAttacks}

A successful attacker is capable of modeling an attack sequence to achieve a desirable effect while remaining undetectable to any on-board fault detection mechanisms. In order to accomplish such stealthy behavior, it is necessary to attain information about critical aspects of the system, such as: acquiring knowledge to the modeled dynamics, sensor measurements, state estimator, and detection procedure(s). To intentionally avoid detection, an intelligent attacker will carefully construct an attack sequence to evade raising any flags. Below we describe a sequence an attacker may take with respect to the Bad-Data detector \cite{BadData} leveraging the chi-square test measure procedure. However, this may be extended to satisfy any detection procedure using a similar concept to avoid detection.

\textit{Zero-alarm attacks} are sequences designed by an attacker that maintains the test measure from exceeding the defined threshold value ($z_k \leq \tau_z$). This class of attack does not trigger an alarm throughout the attack sequence, as the test measure never passes the threshold. In order to satisfy such requirements, an attacker can construct the attack vector by
\begin{equation} \label{eq:zero_alarm_attack}
    \bm{\xi}_k = -\bm{C}\bm{e}_k - \bm{\eta}_k + \bm{\Sigma}^{\frac{1}{2}} \bm{\delta}_k,
\end{equation}
where $\bm{\delta}_k \in \R^s$ is a vector that satisfies $\bm{\delta}_k^{\mathsf{T}}\bm{\delta}_k \leq \tau_z$. With this attack vector designed at a time $k$, the test measure $z_k$ results in
\begin{equation} \label{eq:zero_alarm_attack2}
\begin{split}
    z_k &= (\Tilde{\bm{y}}_k - \bm{C}\hat{\bm{x}}_k)^{\mathsf{T}} \bm{\Sigma}^{-1} (\Tilde{\bm{y}}_k - \bm{C}\hat{\bm{x}}_k) \\
    &= (\bm{C}\bm{e}_k + \bm{\eta}_k + \bm{\xi}_k)^{\mathsf{T}} \bm{\Sigma}^{-1} (\bm{C}\bm{e}_k + \bm{\eta}_k + \bm{\xi}_k) \leq \tau_z,
\end{split}
\end{equation}
that remains within the threshold value to not trigger an alarm. While generating an attack sequence that does not trigger alarms may seem like a favorable attack design, it is necessary to recall that alarms are triggered in a system operating in normal conditions without attacks. If alarms are no longer being triggered as designed for in an attack-free case, then these conditions may raise suspicions of a possible attack. To avoid these alarm rate discrepancies, an attacker would want to design an attack sequence that is undetectable to emulate normal (attack-free) conditions. This class of attack brings us to develop a sequence that exploits the system uncertainties to execute such a malicious attack.

\textit{Hidden attacks} can be defined as designed attack sequences such that alarms are triggered at the same rate as the desired false alarm rate during nominal, attack-free operation. As shown in Fig. \ref{fig:HiddenAttack}, during a hidden attack, a smart attacker can design a sequence where the test measure $z_k$ exceeds the threshold $\tau_z$ at the same rate as nominal conditions. To tune for a desired alarm rate $\alpha$ (in the attack-free scenario) for Bad-Data detection while leveraging the chi-square procedure, the specific threshold $\tau_z$ is found by
\begin{equation} \label{eq:BD_threshold}
    \tau_z = 2\gamma^{-1} \Big( 1 - \alpha, \frac{s}{2} \Big),
\end{equation}
to achieve a desired alarm rate, where $\gamma^{-1}(\cdot,\cdot)$ is the \textit{inverse regularized lower incomplete gamma function} \cite{statsbook}. The vector $\bm{\delta}_k$ from \eqref{eq:zero_alarm_attack} is designed such that
\begin{equation} \label{eq:hidden_attack}
    \PP(z_k > \tau_z) = \PP(\bm{\delta}_k^{\mathsf{T}} \bm{\delta}_k > \tau_z) \approx \alpha,
\end{equation}
to remain hidden from detection.

\begin{figure}[th!b]
\vspace{-5pt}
\centering
\hspace{-7pt} \includegraphics[width=0.43\textwidth]{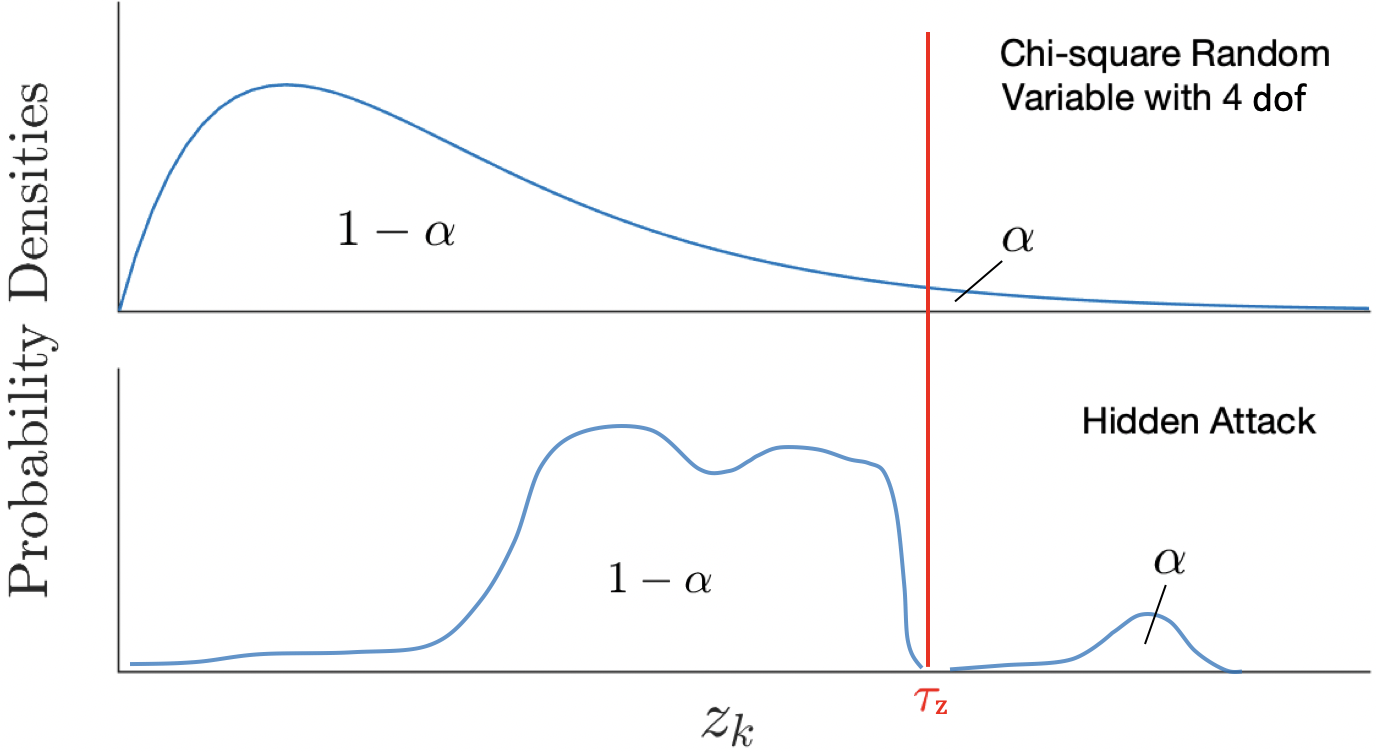}
\vspace{-8pt}
\caption{(top) The chi-square distribution of $z_k$ with $s = 4$ degrees of freedom (dof). A correctly chosen threshold value $\tau_z$ results in the test measure $z_k$ exceeding the threshold at a desired rate of $\alpha$ in the attack-free case. Similarly, during a hidden attack (bottom), an attacker can design an attack sequence such that the alarm rate matches the desired alarm rate $\alpha$, all while altering the distribution of $z_k$ and remaining hidden from detection.}
\label{fig:HiddenAttack}
\end{figure}
\vspace{-10pt}

\subsection{Problem Formulation}
\label{sec:Problem}

An attacker with the specific objective to hijack a system or to degrade system performance, will leave traces of inconsistent behavior. In this work, we focus on deceptive sensor attacks that purposely hide within the noise characteristics of the system model and evade detection of conventional fault detection procedures in order to remain undetected.

\begin{definition}
Sensor measurements are behaving consistently if:
\label{consist_definition}
\begin{itemize}
\item{The test measure follows a chi-square distribution $z_k \sim \chi^2 (s)$ that is determined by the $s$ number of sensors}.
\item{The signed test measure difference switches sign values at a proper (i.e., expected) rate}.
\end{itemize}
\end{definition}

Since we are considering sensor spoofing, an attack vector $\bm{\xi}_k$ containing malicious data can disrupt consistency, thereby causing the test measure to display non-random behavior. Formally, the problem that we are interested in solving is:
\begin{problem}
\label{problem1}
\textit{(Runtime Detection of Measurement Inconsistencies).} Given the quadratic test measure $z_k$, computed from the residual $\bm{r}_k$ as defined in \eqref{eq:Residual} and the inverse of the residual covariance matrix $\bm{\Sigma}$ in \eqref{eq:Residual_Covariance}, find a policy to determine at runtime whether sensor measurements are consistent, i.e., if any condition in Definition \ref{consist_definition} does not hold.
\end{problem}
\begin{section}{Serial Consistency of the Test Measure}
\label{sec:framework}

The overall cyber-physical control system architecture including our detection procedure for serial consistencies is summarized in Fig.~\ref{fig:AttackDiagram}. The monitor is placed in the system feedback to observe the relationship between the sensor measurements and state prediction while leveraging the chi-square test measure. We focus our attention on stealthy sensor attack sequences where a malicious attacker may inject an attack signal to measurements at any point between the sensors and state estimator.

\begin{figure}[th!b]
\vspace{-5pt}
\centering
\includegraphics[width=0.47\textwidth]{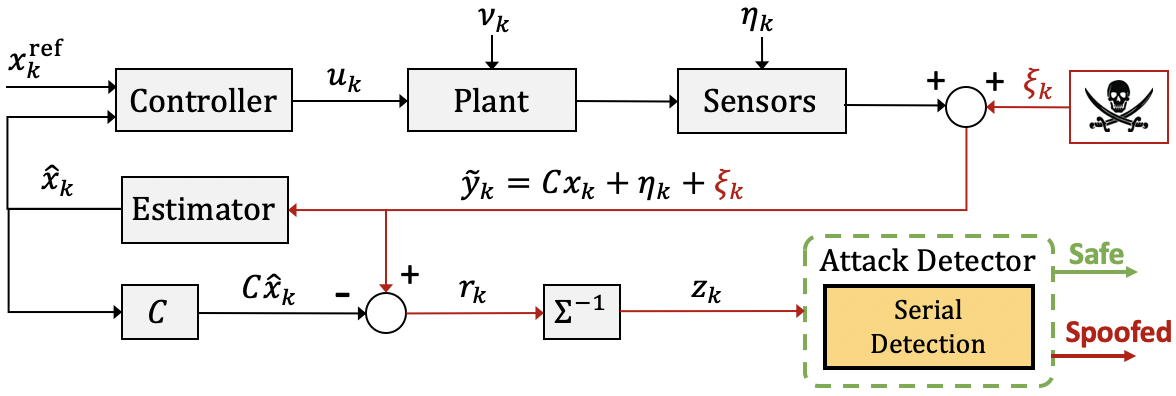}
\vspace{-9pt}
\caption{The architecture of a CPS while experiencing sensor attacks with the Serial Detector placed in the system feedback.}
\label{fig:AttackDiagram}
\end{figure}
\vspace{-15pt}

\subsection{Magnitude-based Detection}
\label{sec:Magnitude_randomness}

The design of the Serial Detector is to find inconsistent behavior of chi-square test measures within its expected distribution due to stealthy attacks to on-board sensor measurements. An attacker deliberately attempting to fool test measure-based detection algorithms may leave traces of inconsistencies within the serial sequence. We propose the Serial Detector that analyzes consecutive chi-square test measures at time instances $k$ and $k-1$, called the \textit{test measure difference}, that is described as:
\begin{equation}
\begin{split} \label{eq:test_measure_difference}
    d_{k} &= z_{k} - z_{k-1}, \\
	& = \bm{r}_k^{\mathsf{T}} \bm{\Sigma}^{-1} \bm{r}_k - \bm{r}_{k-1}^{\mathsf{T}} \bm{\Sigma}^{-1} \bm{r}_{k-1} \in \R.
\end{split}
\end{equation}

\begin{proposition} \label{pro:test_measure_expectation}
    A system that is free from sensor attacks, where we assume consecutive test measures are independent random variables that follow chi-square distributions $z_k, z_{k-1} \sim \chi^2(s)$ with $s$ degrees of freedom, has the following expectations of the test measure difference $d_k$:
    \begin{equation}  \label{eq:test_measure_expectation}
    \begin{split}
        \E[d_{k}] &= \E[z_{k}] - \E[z_{k-1}] = 0, \\
        \mathrm{Var}[d_{k}] &= \mathrm{Var}[z_{k}] + \mathrm{Var}[z_{k-1}] = 4s. 
    \end{split}
    \end{equation}
\end{proposition}
\vspace{5pt}

Given an attack-free system that follows the expectation \eqref{eq:test_measure_expectation} in Proposition \ref{pro:test_measure_expectation}, the test measure difference $d_k$ follows
\begin{equation} \label{eq:variancegamma_dist}
\vspace{-1pt}
    d_k \sim \mathcal{VG} \Big( \E[d_{k}], \sqrt{\mathrm{Var}[d_{k}]}, 0, \frac{2}{s} \Big) \in \R,
    \vspace{-1pt}
\end{equation}
where $\mathcal{VG}( \cdot, \cdot, \cdot, \cdot)$ denotes the \textit{variance-gamma distribution} \cite{seneta_2004}, which is a mixed distribution of the normal distribution and gamma distribution. As the chi-square distribution is a special case of the gamma distribution, the difference of two gamma random variables (i.e. chi-square random variables) results in the variance-gamma distribution \cite{gamma_diff_Klar}. The parameters within the variance-gamma distribution that describe the difference of two chi-square random variables, generalized in \cite{ferrari2019note}, are the location $c = \E[d_{k}]$, spread $\bar{\sigma} = \sqrt{\mathrm{Var}[d_k] }$, asymmetry $\vartheta = 0$, and shape $\lambda = \frac{2}{s}$. The probability density function (PDF) of the variance-gamma distribution follows 
\begin{align} \label{eq:PDF_VG}
\small
    \hspace{-4pt} f(x; c, \bar{\sigma}, \vartheta, \lambda) \hspace{-1pt} =& \frac{ 2e^{\hspace{-.5pt} (\vartheta(x-c)/\bar{\sigma}^2 )} \hspace{-.5pt} |x \hspace{-1pt} - \hspace{-1pt} c|^{ \frac{1}{\lambda} \hspace{-1pt} - \hspace{-1pt} \frac{1}{2}} }{ \bar{\sigma}\sqrt{2 \pi} \lambda^{ \frac{1}{\lambda}} \Gamma \big( \frac{1}{\lambda} \big) } \hspace{-2pt} \bigg( \hspace{-3pt} \frac{1}{\sqrt{2 \bar{\sigma}^2 \hspace{-1pt} / \hspace{-1pt} \lambda \hspace{-1pt} + \hspace{-1pt} \vartheta^2 } } \hspace{-2pt} \bigg)^{ \hspace{-3pt} \frac{1}{\lambda} \hspace{-1pt} - \hspace{-1pt} \frac{1}{2}} \nonumber \\
    & \times K_{ \frac{1}{\lambda} - \frac{1}{2}} \bigg( \frac{ |x-c| \sqrt{2 \bar{\sigma}^2 / \lambda + \vartheta^2 } }{\bar{\sigma}^2} \bigg),
\end{align}
\normalsize
where $K_{ \lambda}$ is the \textit{modified Bessel function of the third kind} of order $\lambda$ and $\Gamma(\cdot)$ is the \textit{gamma function} \cite{statsbook}. During nominal conditions, the test measure difference $d_k$ is a symmetric zero-mean distribution (i.e., parameters $c = \vartheta = 0$). 

\textbf{\textit{False Alarms:}} Similar to other detection algorithms in literature \cite{BadData,CUSUM_Journal}, we leverage an alarm rate to diagnose the health of the system from sensor attacks. The magnitude-based detection scheme compares the test measure difference $d_{k}$ to a threshold $\tau$ by:
\begin{equation}
\begin{split}
    \label{eq:det_thresh}
    \bigg\{ \begin{array}{llrl}
	|d_{k}| > \tau_d & \longrightarrow & \textit{alarm:} & \hspace{-5pt} \zeta_k^M = 1, \\[2pt]
    |d_{k}| \leq \tau_d & \longrightarrow & \textit{no alarm:} & \hspace{-5pt} \zeta_k^M = 0,
    \end{array}
\end{split}
\end{equation} 
where the chosen threshold $\tau_d$ is dependent on the expected test measure difference distribution described in \eqref{eq:variancegamma_dist}. In Fig.~\ref{fig:Var_Gamma_distribution} we show how the distribution of the test measure difference $d_k = z_k - z_{k-1}$ (the difference of two chi-square random variables) is affected by the number of sensors $s$.

\begin{figure}[th!b]
\vspace{-3pt}
\centering
\includegraphics[width=0.45\textwidth]{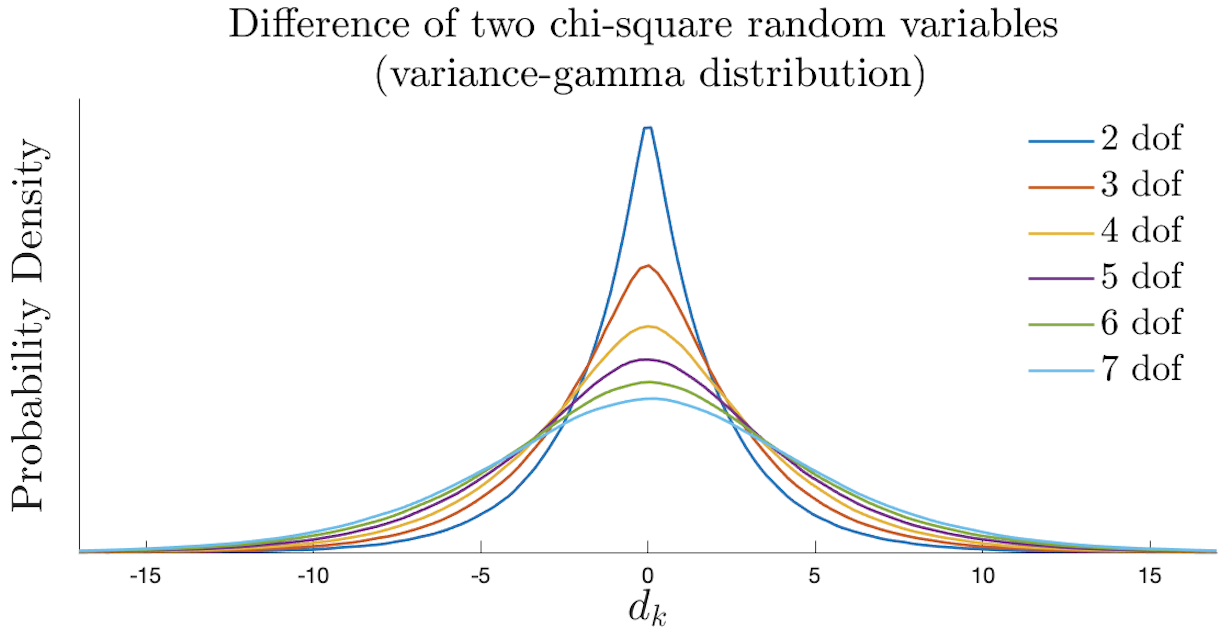}
\vspace{-7pt}
\caption{The resulting distribution of the test measure difference $d_k = z_k - z_{k-1}$ follows a variance-gamma distribution in the attack-free scenario. Shown are the effects on the distribution of $d_k$ for dof $= \{2,3,4,5,6,7\}$.}
\label{fig:Var_Gamma_distribution}
\end{figure}

Under nominal circumstances, i.e. in the absence of attacks, an alarm is triggered at a desired rate $\psi^M_{des} \in (0,1)$ given the chosen threshold value. The following lemma provides a method to choose a threshold to satisfy a user-defined desired alarm rate.
\begin{lemma} \label{lem:VG_choose_threshold}
    Assuming that the system is attack-free (i.e. $\bm{\xi}_k = \bm{0}$) and considering the procedure in \eqref{eq:det_thresh} to trigger an alarm, a specific threshold value $\tau_d$ is chosen by
    \begin{equation} \label{eq:thresh_d}
        \tau_d = \{ \tau_d \in \R_{>0} : \PP( \zeta^M_k = 1 ) = \psi^M_{des} \},
    \end{equation}
    such that the result is a desired alarm rate $\psi^M_{des}$.
\end{lemma}

\begin{proof}
    Let $F_{d_k}(x; c, \bar{\sigma}, \vartheta, \lambda) $ denote the cumulative distribution function (CDF) of the random variable $d_k$ from the PDF in \eqref{eq:PDF_VG}. We compute the inverse CDF for a given desired false alarm rate $\psi^M_{des}$ to find the threshold value
    \begin{equation} \label{eq:magnitude_thresh}
        \tau_d = F_{d_k}^{-1} \Big(1- \frac{\psi^M_{des}}{2}; c, \bar{\sigma}, \vartheta, \lambda \Big) \in \R_{>0},
    \end{equation}
    such that $\PP( \zeta^M_k = 1 ) = \PP(|d_k| > \tau_d) = \psi^M_{des}$ to achieve a desired false alarm rate, thus concluding the proof.
\end{proof}

\textbf{\textit{Alarm Rate Estimation:}} We employ a runtime method of estimating the alarm rate such that we are able to eliminate the need to store a sequence of values. In this work, a Memoryless Runtime Estimator (MRE) \cite{Paul_IFAC} is leveraged to eliminate the need to use a windowed method to compute an alarm rate estimation $\hat{\psi}_{k}^M \in [0,1]$. The MRE algorithm is updated by following
\begin{equation} \label{eq:MRE_algorithm}
    \hat{\psi}_{k}^M = \hat{\psi}_{k-1}^M + \frac{\zeta_{k}^M - \hat{\psi}_{k-1}^M }{\ell},
\end{equation}
where $\ell$ is a user-defined ``pseudo-window" length. The resulting distribution while leveraging MRE can be approximated to a normal distribution for pseudo-window lengths $\ell \geq 10$ \cite{Paul_IFAC} consisting of a variance that follows that of a exponential moving average (EMA) \cite{moving_average}.

\begin{lemma} \label{lem:Magnitude_bounds}
Given the test measure difference $d_k$ defined in \eqref{eq:test_measure_difference} for a system that is assumed to be attack-free and tuned for a desired false alarm rate $\psi^M_{des}$, the estimate alarm rate follows a Normal distribution described by
\begin{equation} \label{eq:Lemma_AR_dist}
    \hat{\psi}_{k}^M \sim \mathcal{N} \bigg( \psi^M_{des}, \frac{\psi^M_{des}(1-\psi^M_{des})}{2\ell-1} \bigg).
\end{equation}
\end{lemma}
\vspace{5pt}

\begin{proof}
We first characterize the magnitude-based detector tuned for a desired false alarm rate $\psi^M_{des}$ as a Binomial distribution $\mathcal{B} \big( \cdot, \cdot \big) $ where $\psi^M_{des}$ is a probability for a ``success" during a specified number of ``trials" (Refer to \cite{statsbook} for further explanations). By way of the binomial approximation for larger pseudo-window size $\ell \geq 10$, a normal distribution can be used to approximate the alarm rate while leveraging MRE \eqref{eq:MRE_algorithm} for estimation that results in
\begin{equation}
    \label{eq:proof_magnitude_bounds}
    \E[\psi^M] = \psi^M_{des} \hspace{-1pt}, \hspace{6pt} \mathrm{Var}[\psi^M] = \frac{ \psi^M_{des}(1 \hspace{-1pt} - \hspace{-1pt} \psi^M_{des})}{2\ell -1}.
\end{equation}
From \eqref{eq:proof_magnitude_bounds} we obtain the distribution in \eqref{eq:Lemma_AR_dist} to characterize the estimated alarm rate for magnitude-based detection.
\end{proof}

With the known expected false alarm rate distribution described in \eqref{eq:Lemma_AR_dist}, we want to find bounds on the estimated alarm rate $\hat{\psi}_k^M, \; \forall k$ to determine if an attack has occurred. The following corollary provides detection bounds with a specific level of confidence $1-\beta$, where $\beta \in [0,1]$ is a user defined level of significance\footnote{Reducing the value of $\beta$ causes the detection bounds to move farther from the expected alarm rate, thus reducing the frequency of falsely ``detecting" under nominal (i.e., no attack) conditions while consequently giving an attacker more freedom to design an attack without being detected, while the opposite is true when increasing $\beta$.}.

\begin{corollary}
    Assuming a system with $s$ sensors that employs the chi-square detection scheme that is monitoring the test measure difference \eqref{eq:test_measure_difference} with a level of significance $\beta$ while leveraging MRE \eqref{eq:MRE_algorithm}, detection of sensor attacks occurs when $\Omega_{-} \leq \hat{\psi}_{k}^M \leq \Omega_{+}$ is not satisfied where
    \begin{equation} \label{eq:lemma_magnitude_bounds}
    \Omega_{\pm} = \E[\psi^M] \pm Z \sqrt{\frac{ \E[\psi^M] (1-\E[\psi^M])}{2\ell -1}}.
    \end{equation}
\end{corollary}
\vspace{2pt}

\begin{proof}
    We construct confidence intervals for a normally distributed variable of a specific confidence level, determined by z-score $Z = \big| \Phi^{-1} \big(\frac{\beta}{2} \big) \big|$, that provide detection bounds by
    \begin{equation} \label{eq:lemma_magnitude_bounds2}
    \begin{split}
    \small
    & \E[\psi^M] - \Big| \Phi^{-1} \Big(\frac{\beta}{2} \Big) \Big| \sqrt{\frac{ \E[\psi^M] (1-\E[\psi^M])}{2\ell -1}} \leq \hat{\psi}_{k}^M \\
    & \leq \E[\psi^M] + \Big| \Phi^{-1} \Big(\frac{\beta}{2} \Big) \Big| \sqrt{\frac{ \E[\psi^M] (1-\E[\psi^M])}{2\ell -1}}
    \end{split}
    \end{equation}
    which satisfy \eqref{eq:lemma_magnitude_bounds}, concluding the proof.
\end{proof}

Detection of sensor attacks occur when an estimated alarm rate $\hat{\psi}_{k}^M$ travels beyond the thresholds from $\Omega_{\pm} = [ \Omega_-,\Omega_+ ]$.

\subsection{Signed Randomness}
\label{sec:Sign_randomness}

To further strengthen detection of inconsistencies within the test measure, we monitor the ``runs" behavior of the test measure difference sequence. While a smart attacker may be able to fool the magnitude-based monitor as discussed in Section \ref{sec:Magnitude_randomness}, an attacker may leave traces of non-random behavior on the signed test measure difference. The test we use to monitor for signed randomness is influenced by the Serial Independence Runs (SIR) Test \cite{serial_test}. An example of the SIR test is shown in Fig. \ref{fig:Serial_Randomness}, where it monitors a sequence of data by first computing the difference between the current and previous data values and taking the sign of the difference to create a two-valued data sequence (i.e., positive and negative values). Then, the number of observed runs $N_r$, defined as consecutive values of the same sign, are counted over the sequence length. In Fig. \ref{fig:Serial_Randomness} we see that over the sequence length $W=14$ there are $N_r = 12$ runs, which are highlighted by the red and blue lines.

\begin{figure}[th!b]
\vspace{-4pt}
\centering
\includegraphics[width=0.45\textwidth]{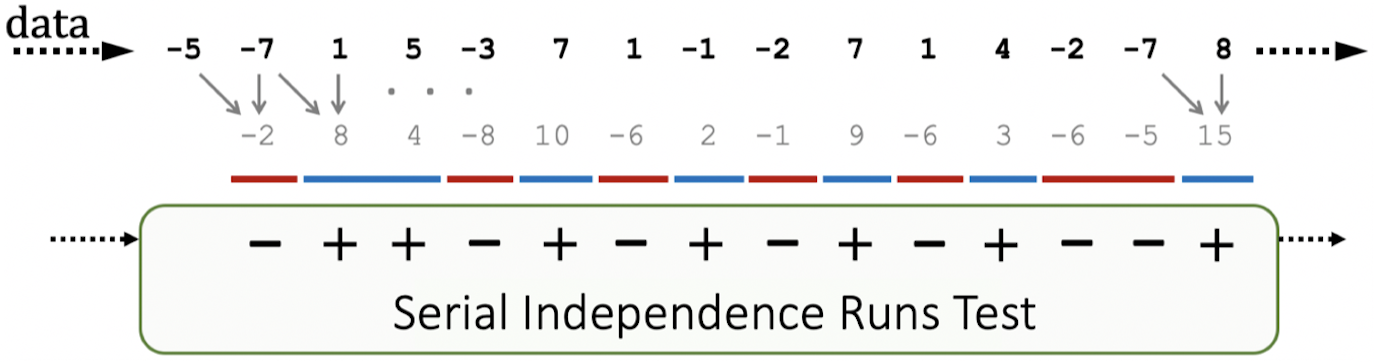}
\vspace{-7pt}
\caption{A sequence of data, from left to right, converted to sequence of signed values while leveraging the Serial Independence Runs Test.}
\label{fig:Serial_Randomness}
\vspace{-6pt}
\end{figure}

A drawback of the SIR test is the requirement to store the $W$ length sequence of test measure differences $d_k$ and then count the number of observed runs $N_r$ over this sequence. Alternatively, we would like to use a window-less method to eliminate the need for storing an entire sequence to determine whether the signed test measure difference is behaving randomly. To this end, we propose to observe sign switches at runtime by triggering an alarm at a time $k$ when the the present test measure difference is of the opposite sign from the previous test measure difference at time $k-1$.

We first compute the sign of the test measure difference by the following:
\begin{equation}
\vspace{-1.5pt}
\begin{split}
    \label{eq:Sign_function2}
    \mathrm{sgn}(d_{k}) := \left\{ \begin{array}{rl}
	-1, & \hspace*{2pt} \text{if } d_{k} < 0, \\
    0, & \hspace*{2pt} \text{if } d_{k} = 0, \\
    1, & \hspace*{2pt} \text{if } d_{k} > 0, \\
    \end{array} \right.
\end{split}
\end{equation} 
and given that the distribution of the test measure difference $d_k$ is symmetric (assuming nominal conditions) over the expected value $\E[d_k] = 0$, the probability of observing the signed values of the test measure difference are
\begin{equation}
\label{eq:binomial_probs}
    \begin{split}
    \PP\big(\text{sgn}(d_k) = -1 \big) &= 0.5, \\
    \PP\big(\text{sgn}(d_k) = 0 \big) &= 0, \\
    \PP\big(\text{sgn}(d_k) = 1 \big) &= 0.5.
    \end{split}
\end{equation}

As sensor measurements are received at every $k$th time instance, the next test measure difference $d_k$ is computed from \eqref{eq:Residual}, \eqref{eq:test_measure}, and \eqref{eq:test_measure_difference}. A switch of the test measure difference sign signifies the end of a run and an alarm $\zeta_{k}^S \in \{0,1\}$ is triggered such that $\zeta_{k}^S = 1$ at a time instance $k$, otherwise $\zeta_{k}^S = 0$. The procedure to trigger a test measure difference alarm follows: 
\begin{equation} \label{pro:TestMeasure_Diff_Sign_Switch}
\begin{split}
    \zeta_{k}^S := \left\{ \begin{array}{rl}
	1, & \hspace*{2pt} \text{if } \mathrm{sgn}(d_{k}) = -\mathrm{sgn}(d_{k-1}), \\
    0, & \hspace*{2pt} \text{otherwise}. \\
    \end{array} \right.
\end{split}
\end{equation} 

The alarm $\zeta_{k}^{S} \in \{0,1\}$ in \eqref{pro:TestMeasure_Diff_Sign_Switch} is then sent into the MRE to provide an updated runtime estimate of the test measure difference alarm rate $\hat{\psi}_{k}^{S}$ at time instance $k$.

\begin{lemma} \label{lem:SI}
Given a system that is not experiencing attacks, the test measure difference alarm rate while leveraging MRE \eqref{eq:MRE_algorithm} for estimation is described as a Normally distributed random variable by the following:
\vspace{-4pt}
\begin{equation}
\label{eq:Serial_SignChange}
    \hat{\psi}_{k}^S \sim \mathcal{N} \big( \E[\psi^S] , \mathrm{ Var}[\psi^S] \big).
\end{equation}
\end{lemma}
\vspace{3.5pt}

\begin{proof} \label{pr:SI_dist}
We want to convert the distribution of expected runs $\E[N_R] $ of test measure differences over a window-based sequence of length $W$ described in \cite{serial_test} by
\begin{equation} \label{eq:Serial_Independence_Runs}
    \vspace{-1pt}
    N_R \sim \mathcal{N} \bigg( \frac{2W-1}{3}, \frac{16W-29}{90} \bigg),
    \vspace{-1pt}
\end{equation}
to a runtime rate of expected test measure difference sign switching $\E[\psi^S]$. By first obtaining the asymptotic distribution and then transforming the expected observed runs to an expected rate of observed alarms $\E[\psi^S] = \frac{\E[N_R]}{W}$, we arrive to the expected sign switching alarm rate distribution
    \begin{equation} \label{eq:Serial_SignChange_MeanVar}
    \E[\psi^S] = \frac{2}{3}, \hspace*{5pt} \mathrm{ Var}[\psi^S] = \frac{16}{90(2\ell-1)},
    \end{equation}
while leveraging MRE for window-less estimation.
\end{proof}

The following corollary provides a proof for detection bounds of $\hat{\psi}_{k}^S$ to satisfy an expected alarm rate $\E[\psi^S]$.
\begin{corollary} \label{cor:SI_Bounds}
Given the test measure differences $d_{k} = z_{k} - z_{k-1} $, detection occurs by the test measure difference alarm rate when $ \Psi_- \leq \hat{\psi}_{k}^S \leq \Psi_+$ is not satisfied where
\begin{equation} 
\label{eq:SI_corollary}
    \Psi_{\pm} = \pm \Big| \Phi^{-1} \Big( \frac{\beta}{2} \Big) \Big| \sqrt{\frac{16}{90(2\ell-1)}} + \frac{2}{3}. 
\end{equation}

\vspace{1pt}
\end{corollary}

\begin{proof}
For a desired level of significance $\beta$ we find the bounds of $ \hat{\psi}_{k}^S$ for an expected alarm rate $\E[\psi^S]$ are
\begin{equation}\label{eq:SI_bound_proof}
- Z  \sqrt{ \frac{16}{90(2\ell-1)} } + \frac{2}{3} \leq \hat{\psi}_{k}^S \leq Z \sqrt{\frac{16}{90(2\ell-1)}} + \frac{2}{3},
\end{equation}
where z-score is $Z = \big| \Phi^{ -1} \big( \frac{\beta}{2} \big)  \big|$. From \eqref{eq:SI_bound_proof} we can finally obtain the detection bounds of $\Psi_{\pm}$ in \eqref{eq:SI_corollary} for alarm triggering at an expected alarm rate $\E[\psi^S] $.
\end{proof}

We should note that while we describe a runtime method for detecting anomalous signed behavior within the serial sequence of a chi-square random variable $z_k$, this technique may be used on any randomly distribution variable. As this method is non-parametric, the signed behavior is independent from its underlying distribution \cite{serial_test,Seigel}.
\vspace{-4pt}

\end{section}
\begin{section}{Undetectable Attacks}
\label{sec:Undetected_Attacks}

This section analyzes the attack sequence that an attacker must make in order to remain undetected from our serial randomness-based detector. Continuing with assumptions previously made in Section \ref{sec:HiddenAttacks}, we assume a worst-case scenario where a smart attacker has access to the system model, noise characteristics, control inputs, and state estimator to fool our detection technique. In particular, we focus on the attack sequences of $\bm{\xi}_k$ that can disrupt nominal closed-loop system behavior while remaining hidden. 

\subsection{Magnitude-based Detection}
\label{sec:mag_based_detection}

We begin by considering an attack sequence that does not allow the magnitude-based alarm rate $\hat{\psi}_k^M$ to travel beyond detection bounds described in \eqref{eq:lemma_magnitude_bounds}. If we recall the test measure difference $d_k$ in \eqref{eq:test_measure_difference}, but written in terms of the sensor attack vector $\bm{\xi}_k$, we have
\small
\begin{align} \label{eq:test_measure_diff_AV}
    d_{k} =& \hspace{2pt} \bm{r}_k^{\mathsf{T}} \bm{\Sigma}^{-1} \bm{r}_k - \bm{r}_{k-1}^{\mathsf{T}} \bm{\Sigma}^{-1} \bm{r}_{k-1} \nonumber \\
    =& \hspace{2pt} (\bm{Ce}_k + \bm{\eta}_k + \bm{\xi}_k)^{\mathsf{T}} \bm{\Sigma}^{-1} (\bm{Ce}_k + \bm{\eta}_k + \bm{\xi}_k) \\
    &- \hspace{-1pt} (\bm{Ce}_{k-1} + \bm{\eta}_{k-1} + \bm{\xi}_{k-1})^{\mathsf{T}} \bm{\Sigma}^{-1} (\bm{Ce}_{k-1} + \bm{\eta}_{k-1} + \bm{\xi}_{k-1}) \nonumber \\
    =& \hspace{2pt} (\bm{Ce}_k + \bm{\eta}_k + \bm{\xi}_k)^{\mathsf{T}} \bm{\Sigma}^{-1} (\bm{Ce}_k + \bm{\eta}_k + \bm{\xi}_k) - z_{k-1}. \nonumber
\end{align}
\normalsize

In order for an attacker to not trigger the alarm $\zeta_k^M = 1$ at time $k$, i.e. a zero-alarm attack, the sensor attack vector must maintain the test measure difference to satisfy $|d_k| \leq \tau_d$. For an attack vector sequence and the designed variance-gamma distribution threshold $\tau_d$, we define a suitable vector
\begin{equation} \label{eq:suitable_vector}
    \bm{\delta}_k = \{ \bm{\delta}_k \in \R^s : | \bm{\delta}_k^{\mathsf{T}} \bm{\delta}_k - z_{k-1} | \leq \tau_d \},
\end{equation}
that leads to the test measure difference $d_k$ not triggering an alarm. Therefore, for any time $k$, the attack vector follows
\begin{equation} \label{eq:mag_based_attack_vector}
    \bm{\xi}_k = - \bm{Ce}_k - \bm{\eta}_k + \bm{\Sigma}^{\frac{1}{2}} \bm{\delta}_k, 
\end{equation}
where $\bm{\Sigma}^{\frac{1}{2}}$ is the symmetric square root of the residual covariance matrix \eqref{eq:Residual_Covariance}, such that 
\begin{equation}
    |d_k| = |z_k - z_{k-1}| \leq \tau_d,
\end{equation}
is satisfied. To remain hidden from detection, an attacker must trigger alarms at a rate which the system is expecting. For the case of hidden attacks to evade detection of our serial monitor for magnitude-based detection, a suitable vector $\bm{\delta}_k$ in \eqref{eq:suitable_vector} is constructed as
\begin{equation} \label{eq:hidden_suitable vector}
    \PP( |d_k| > \tau_d ) = \PP( | \bm{\delta}_k^{\mathsf{T}} \bm{\delta}_k - z_{k-1} | > \tau_d) \approx \psi^M_{des},
\end{equation}
to emulate the alarm rate that would be seen during nominal conditions. More specifically, an observed estimated alarm rate computed in \eqref{eq:MRE_algorithm} must remain within detection bounds found in \eqref{eq:lemma_magnitude_bounds} to remain undetected. To ensure detection does not occur for the magnitude-based alarm rate, an attacker must design the attack vector such that the alarm rate remains within detection bounds, $\hat{\psi}_k^M \in [\Omega_-,\Omega_+]$. To remain below the upper detection bound, the vector $\bm{\delta}_k$ follows
\begin{equation} \label{eq:worstcase_magnitude_alarm_upper}
    | \bm{\delta}_k^{\mathsf{T}} \bm{\delta}_k - z_{k-1} | \leq \tau_d \hspace{1pt} : \hspace{2pt} \Big( \Omega_+ - \hat{\psi}_{k-1}^M - \frac{1 - \hat{\psi}_{k-1}^M }{\ell} \Big) < 0,
\end{equation}
to guarantee $\hat{\psi}_k^M \leq \Omega_+$. Additionally, a requirement to remain above the lower detection bound adheres to
\begin{equation} \label{eq:worstcase_magnitude_alarm_lower}
    | \bm{\delta}_k^{\mathsf{T}} \bm{\delta}_k - z_{k-1} | > \tau_d \hspace{1pt} : \hspace{2pt} \Big( \Omega_- - \hat{\psi}_{k-1}^M + \frac{ \hat{\psi}_{k-1}^M }{\ell} \Big) > 0.
\end{equation}

\subsection{Sign-based Detection}
\label{sec:Sign_based_detection}

We continue with a scenario for an attacker to evade detection from the Serial Detector in which the attack design is required to satisfy signed randomness throughout the sequence. Similar to the magnitude-based detection in Section \ref{sec:mag_based_detection}, the attack sequence must result in alarm rates that emulate attack-free conditions to remain hidden from detection. To achieve this, the sign-based alarm rate satisfies
\begin{equation} \label{eq:expected_sign_switching}
\begin{split}
    & \hspace{-4pt} \PP \big( \mathrm{sgn}(d_k) = -\mathrm{sgn}(d_{k-1}) \big) = \\
    & \hspace{-4pt} \PP \big( \mathrm{sgn}(z_k \hspace{-1pt} - \hspace{-1pt} z_{k-1}) = -\mathrm{sgn}(z_{k-1} \hspace{-1pt} - \hspace{-1pt} z_{k-2}) \big) \hspace{-1pt} \approx \E[\psi^S],
\end{split}
\end{equation}
in order to behave similarly to nominal conditions. In order to not cause a sign switching condition, i.e. signed-based alarm $\zeta_k^S = 0$, the sign of $d_k$ must consist of the same sign as $d_{k-1}$. In terms of the vector $\bm{\delta}_k$ while leveraging \eqref{eq:zero_alarm_attack}, the following inequality
\begin{equation}
\begin{split} \label{eq:vector_no_switching}
    \left\{ \begin{array}{ll}
	\bm{\delta}_k^{\mathsf{T}} \bm{\delta}_k > z_{k-1}, & \hspace*{2pt} \text{if } d_{k-1} > 0, \\[1pt]
    \bm{\delta}_k^{\mathsf{T}} \bm{\delta}_k < z_{k-1}, & \hspace*{2pt} \text{if } d_{k-1} < 0, \\
    \end{array} \right.
\end{split}
\end{equation} 
must be satisfied to not cause a sign change, thus not triggering an alarm. If the signed component alarm rate $\psi_k^S, \; \forall k$ approaches the upper detection bound, the following equation guarantees $\hat{\psi}_k^S \leq \Psi_+$, where
\begin{equation} \label{eq:worstcase_sign_alarm_upper}
\small
    \mathrm{sgn}(\bm{\delta}_k^{\mathsf{T}} \bm{\delta}_k - z_{k-1}) = \mathrm{sgn}(d_{k-1}) \hspace{1pt} : \hspace{2pt} \Big( \Psi_+ - \hat{\psi}_{k-1}^S - \frac{1 \hspace{-1pt} - \hspace{-1pt} \hat{\psi}_{k-1}^S }{\ell} \hspace{-1pt} \Big) \hspace{-1pt} < 0,
\end{equation}
\normalsize
thus maintaining the alarm rate within bounds. Similarly, the requirement to not cross below the lower bound adheres to
\begin{equation} \label{eq:worstcase_sign_alarm_lower}
\small
    \mathrm{sgn}(\bm{\delta}_k^{\mathsf{T}} \bm{\delta}_k - z_{k-1}) = -\mathrm{sgn}(d_{k-1}) \hspace{1pt} : \hspace{2pt} \Big( \Psi_- - \hat{\psi}_{k-1}^S \hspace{-1pt} + \hspace{-1pt} \frac{ \hat{\psi}_{k-1}^S }{\ell} \hspace{-1pt} \Big) \hspace{-1pt} > 0,
\end{equation}
\normalsize
to remain undetectable from the Serial Detector.

\end{section}
\begin{section}{Results} \label{sec:Results}

The proposed Serial Detector was validated in simulation and compared to state-of-the-art detection techniques: Bad-Data (BD) \cite{BadData}, Cumulative Sum (CUSUM) \cite{CUSUM_Journal}, and Cumulative Sign (CUSIGN) \cite{Paul_IFAC} detectors. 
The case study presented in this paper is an autonomous differential-drive UGV with the following linearized model \cite{vehiclemodel}
\begin{equation}
\begin{split}
\label{eq:UGV_dynamics}
    \dot{v} &= \frac{1}{m}(F_l+F_r-B_rv), \\
    \dot{\omega} &= \frac{1}{I_z}\Big(\frac{w}{2}(F_l-F_r)-B_l\omega \Big), \text{ } \dot{\theta} = \omega,
\end{split}
\end{equation}
where $v$, $\theta$, and $\omega$ denote velocity, vehicle heading angle, and angular velocity, forming the state vector $\bm{x} = [v,\theta,\omega]^{\mathsf{T}}$. $F_l$ and $F_r$ describe the left and right input forces from the wheels, $w$ is the vehicle width, while $B_r$ and $B_l$ are resistances due to the wheels rolling and turning. Two sensors ($s = 2$) receive measurements of the states $x_1 = v$ and $x_2 = \theta$ with a sampling rate $t_s = 0.01$.

We perform two different attack sequences: \textit{Bias Attack} where the attack sequence concentrates the test measure distribution such that the magnitude detectors (BD and CUSUM) trigger alarms at a desired rate while signed behavior monitored by CUSIGN remains consistent whereas a \textit{Pattern Attack} creates patterned concentrations on the chi-square test measure difference $d_{k}$. In Fig. \ref{fig:Attack_dists}, the resulting distributions for each case are shown that include: (a) the \textit{No Attack} case where $z_k \sim \chi^2(s=2)$, (b) \textit{Bias Attack}, and (c) \textit{Pattern Attack}. 
Both the Bad-Data and CUSUM detectors are tuned for a desired alarm rate of $\alpha = 0.20$ (see \cite{BadData,CUSUM_Journal}) and the CUSIGN detector has an expected alarm rate of $0.0833$ (see \cite{Paul_IFAC}). The magnitude component of our proposed Serial Detector is tuned for an expected alarm rate $ \E[\psi^M] = 0.20$ and the expected alarm rate for the signed component is $\E[\psi^S] = \frac{2}{3}$. All detectors employ detection bounds that are $3$ standard deviations from their expectation.

\begin{figure}[htb!]
\vspace{-3pt}
\begin{tabular}{ccc}
\hspace{-10pt} \subfigure[\label{fig:dist1} ]{\includegraphics[width = 0.155\textwidth]{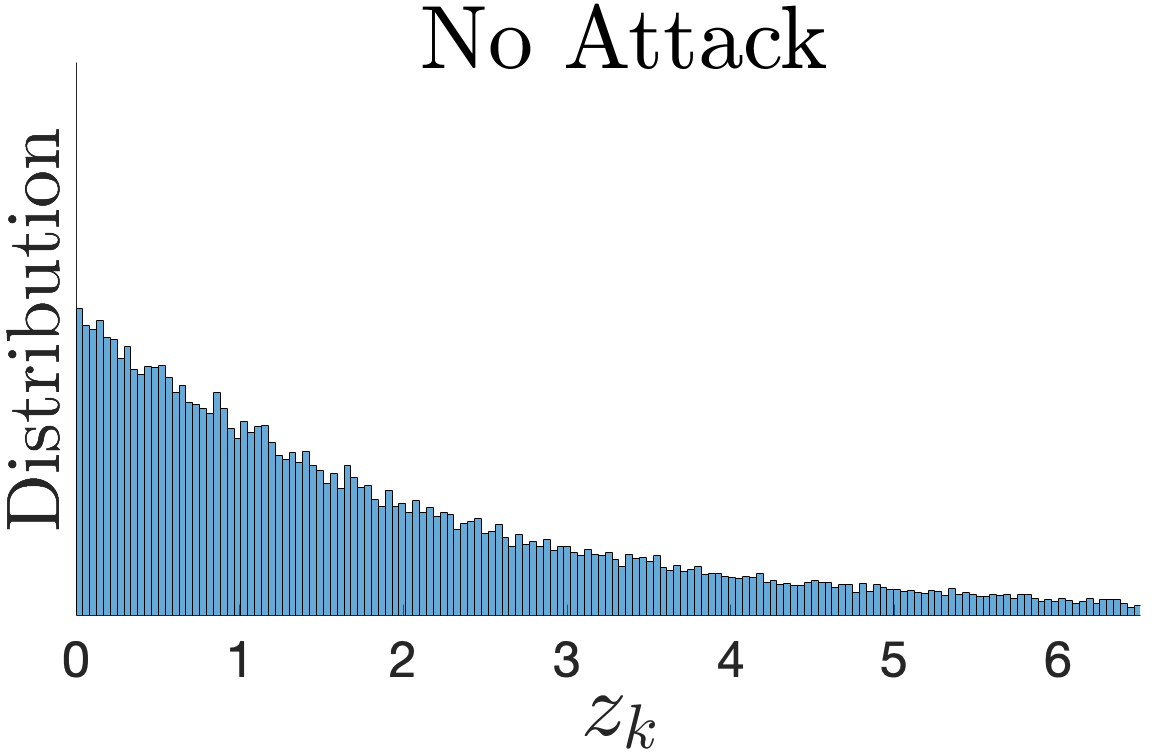}} &
\hspace{-10pt} \subfigure[\label{fig:dist2} ]{\includegraphics[width = 0.155\textwidth]{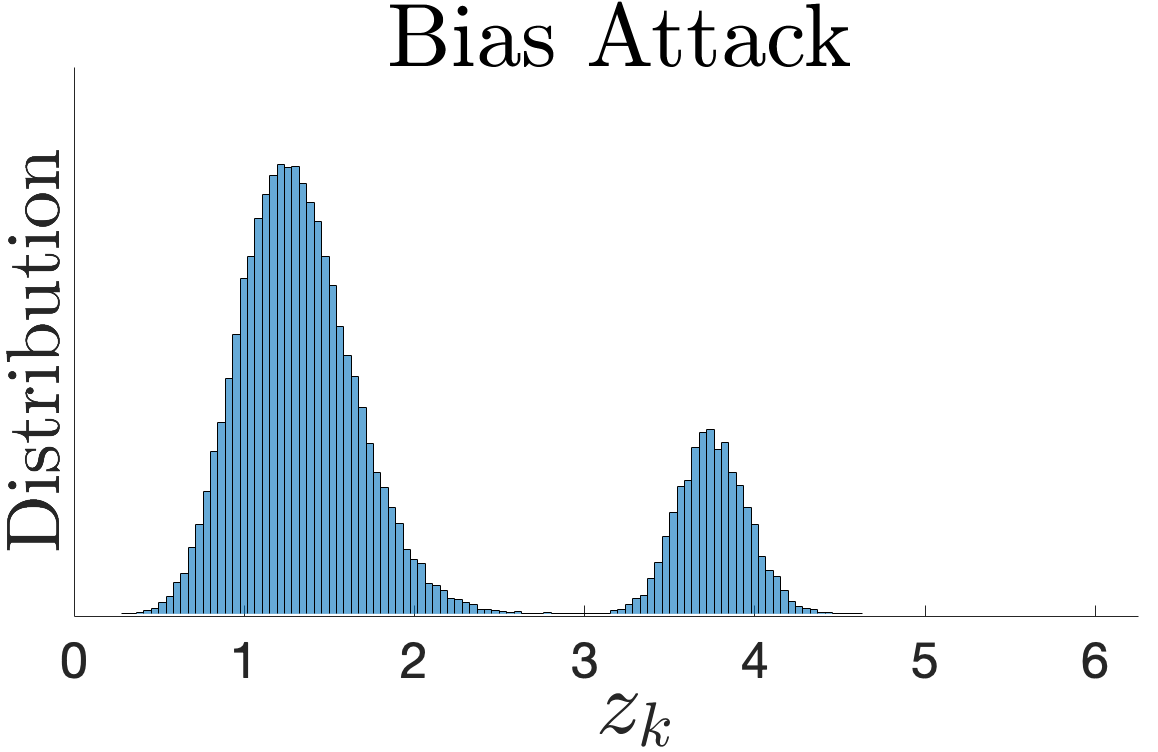}} &
\hspace{-10pt} \subfigure[\label{fig:dist3} ]{\includegraphics[width = 0.155\textwidth]{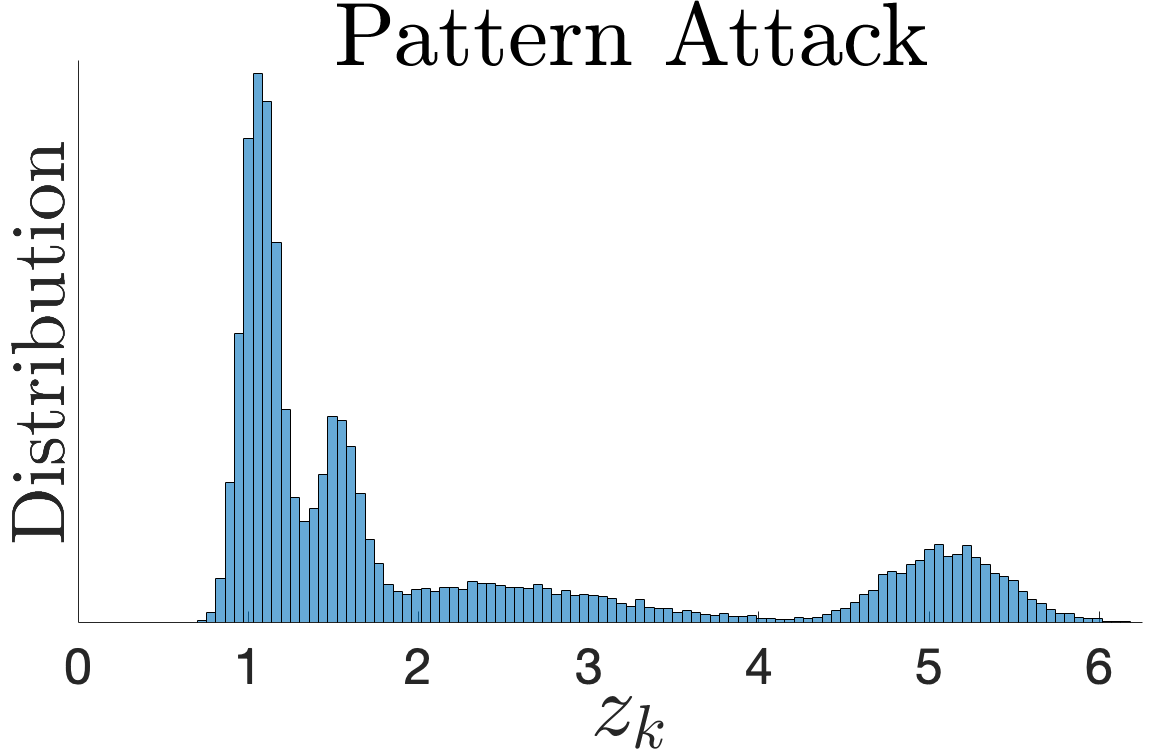}}
\end{tabular}
\vspace{-12pt}
\caption{The test measure $z_k$ distributions when (a) \textit{No attack}, (b) \textit{Bias Attack}, and (c) \textit{Pattern Attack} occur in the simulation case study.}
\label{fig:Attack_dists}
\end{figure}
\vspace{-6pt}

Next, we include a simulation showing the detector alarm rates during \textit{No Attack} at times $k < 20000$, \textit{Bias Attack} at $20000 \leq k < 40000$ and \textit{Pattern Attack} beginning at time step $k \geq 40000$. During the \textit{Bias Attack} in Fig. \ref{fig:CompleteAttack}, the attack fools the BD, CUSUM, and CUSIGN detectors, but the magnitude component of the serial monitor notices the change in the test measure sequence due to the attack. The sign component does not detect the attack, as a bias attack does not disrupt the change of signed behavior of the test measure difference $d_k$. The \textit{Pattern Attack}, while preserving expected test measure difference magnitude behavior, interferes with the expected sign switching rate of the test measure difference. As expected, in the absence of sensor attacks where $k < 20000$, alarm rates for all detection procedures have distributions centered at their expectations. While these modeled attack sequences are primitively designed examples that can fool comparative detectors (e.g. BD, CUSUM, and CUSIGN detectors), the Serial Detector is able to exploit hidden behaviors to strengthen detection capabilities.

\begin{figure}[htb!]
\vspace{-3pt}
\begin{tabular}{c}
\hspace{-5pt}\subfigure[\label{fig:mag_based}]{\includegraphics[width = 0.48\textwidth]{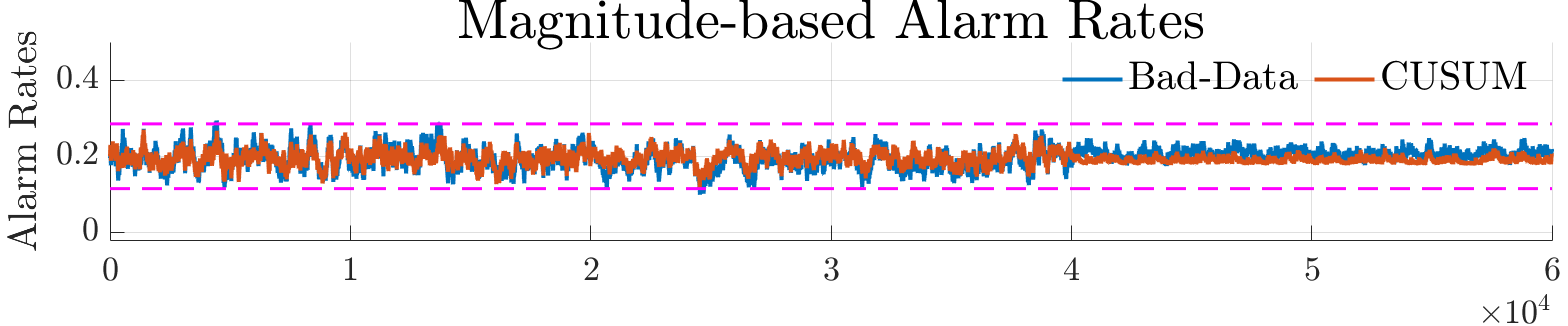}} \\[-2pt]
\hspace{-5pt}\subfigure[\label{fig:rand_based}]{\includegraphics[width = 0.48\textwidth]{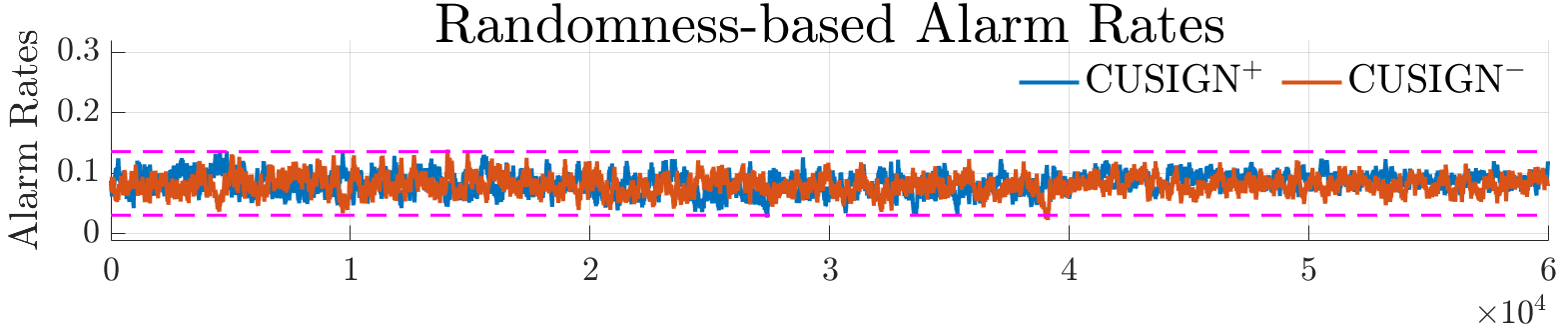}} \\[-2pt]
\hspace{-5pt}\subfigure[\label{fig:serial}]{\includegraphics[width = 0.48\textwidth]{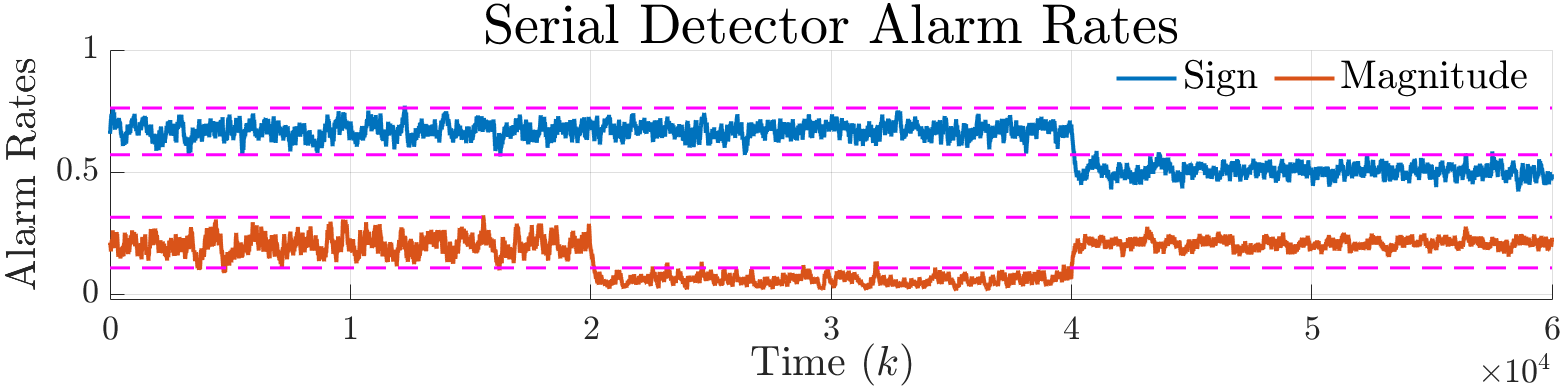}}
\end{tabular}
\vspace{-13pt}
\caption{Resulting alarm rates during the \textit{No Attack}, \textit{Pattern Attack}, and \textit{Bias Attack}. During both the attack scenarios, the comparable detectors (BD, CUSUM, and CUSIGN) are fooled, while the magnitude and sign components of the Serial Detector discover the \textit{Bias} and \textit{Pattern} attacks. Dashed magenta lines represent $3\sigma$ detection bounds for each detector.}
\label{fig:CompleteAttack}
\end{figure}
\vspace{-10pt}

\end{section}
\begin{section}{Conclusions} 
\label{sec:conclusion}

In this paper we have proposed the Serial Detector to discover inconsistent test measure behavior due to hidden cyber-attacks while employing a chi-square detection procedure. Our detection approach monitors the magnitude and signed sequence of the test measure differences to detect inconsistent behavior. We characterized the expected alarm rates for both magnitude and sign, which are dependent on the system model. Furthermore, we provide bounds on detection while also providing an analysis of the detection bounds of our scheme. The proposed approach was validated through simulations on a UGV case study. While our proposed Serial Detector can not replace traditional test measure-based detection schemes, however, it can provide another layer of security to detect hidden attacks that are deceptive to these state-of-the-art detectors.

\section*{Acknowledgments} 
This work is based on research sponsored by ONR under agreement number N000141712012, and NSF under grant \#1816591.

\end{section}

\bibliographystyle{IEEEtran}
\bibliography{ms}

\end{document}